\documentclass[a4paper,11pt]{amsart}

\usepackage{url}
\usepackage{cite}

\usepackage{amssymb}
\usepackage{color}

\usepackage[normalem]{ulem}
\usepackage[utf8]{inputenc}
\usepackage{graphicx}

\usepackage{soul}
\usepackage[normalem]{ulem}

\allowdisplaybreaks

\def\undersetbrace#1\to#2{\underbrace{#2}_{#1}}
\def\oversetbrace#1\to#2{\overbrace{#2}^{#1}}
\def\AMSunderset#1\to#2{\underset{#1}{#2}}
\def\AMSoverset#1\to#2{\overset{#1}{#2}}

\newcommand{\nmb}[2]{\ifx!#1{\ref{nmb:#2}}%
\else\if.#1{\label{nmb:#2}}%
\else\if0#1{\label{nmb:#2}}%
\else{{#2}}%
\fi\fi\fi
}

\def\o{\circ}
\def\al{\alpha}
\def\ga{\gamma}

\def\ep{\varepsilon}

\def\la{\lambda}
\def\rh{\rho}

\def\om{\omega}
\def\Ga{\Gamma}

\def\i{^{-1}}
\def\x{\times}

\let\on=\operatorname

\newcommand{\tlinpsik}{\tilde\linpsi_I}
  \newcommand{\chihere}{\red{\psi}}

 \newcommand{\btheta}{\blue{\underline{\theta}}}
 \newcommand{\bvarphi}{\blue{\varphi}}
\newcommand{\scalarfield}{\blue{ {\Phi}}}

\newcommand{\zchi}{{\mathring{\chi}}}
\newcommand{\zomega}{{\mathring{\omega}}}
\newcommand{\zu}{{\mathring{u}}}
\newcommand{\PV}{P_{\mathcal V}}
\newcommand{\Pzc}{\mathrm{P}_{{\zu }}}
\newcommand{\Pzperp}{\mathrm{P}_{{\zu }^\perp}}

\newcommand{\spacenabla}{{D}}

\newcommand{\emm}{\red{m}}

\newcommand{\linpsi}{\red{\psi}}
\newcommand{\linphi}{\red{\phi}}

\newcommand{\Vpp}{\blue{  G'(0)} }

\newcommand{\pVpp}{\red{+V^2\Vpp}}

\newcommand{\elambda}{\red{\sigma}}

\newcommand{\zh}{{\mathring{h}}}
\newcommand{\zP}{{\mathring{P}}}
\newcommand{\zglor}{{\mathring{\glorentz}}}
\newcommand{\lambdaell}{\red{\lambda_I}}
\newcommand{\varphiell}{\red{\varphi_I}}
\newcommand{\linpsiell}{\red{\linpsi_I}}

\newcommand{\bean}{\begin{eqnarray}\nn}

\newcommand{\nn}{\nonumber}

 \newcommand{\mcV}{{\mycal V}}

\newcommand{\zg }{{ \mathring{g}}}
\newcommand{\zV }{{ \mathring{V}}}

\newcommand{\glorentz}{{ {\mathbf g}}}
\newcommand{\nablariem}{{ {D}}}

\newcommand{\griem}{{ {\mathfrak g}}}

\newcommand{\loc}{\red{\mathrm{loc}}}

\newcommand{\blue}[1]{{\color{blue} #1}}
\newcommand{\red}[1]{{\color{red} #1}}

\DeclareFontFamily{OT1}{rsfs}{}
\DeclareFontShape{OT1}{rsfs}{m}{n}{ <-7> rsfs5 <7-10> rsfs7 <10->
rsfs10}{} \DeclareMathAlphabet{\mathscr}{OT1}{rsfs}{m}{n}

\newcommand{\eq}[1]{\eqref{#1}}

\newcommand{\bel}[1]{\begin{equation}\label{#1}}
\newcommand{\beal}[1]{\begin{eqnarray}\label{#1}}
\newcommand{\beadl}[1]{\begin{deqarr}\label{#1}}
\newcommand{\eeadl}[1]{\arrlabel{#1}\end{deqarr}}
\newcommand{\eeal}[1]{\label{#1}\end{eqnarray}}
\newcommand{\eead}[1]{\end{deqarr}}
\newcommand{\eea}{\end{eqnarray}}
\newcommand{\eeaa}{\end{eqnarray*}}

\newcommand{\be}{\begin{equation}}
\newcommand{\ee}{\end{equation}}

\DeclareFontFamily{OT1}{rsfs}{}
\DeclareFontShape{OT1}{rsfs}{m}{n}{ <-7> rsfs5 <7-10> rsfs7 <10->
rsfs10}{} \DeclareMathAlphabet{\mycal}{OT1}{rsfs}{m}{n}

\newcommand{\mcL}{{\mycal L}}
\newcounter{mnotecount}[section]

\newcommand{\N}{{\mathbb N}}
\newcommand{\Z}{{\mathbb Z}}

\newcommand{\rmnote}[1]{}

\def\mysavedown#1{\edef\mysubs{\mysubs#1}}
\def\mysaveup#1{\edef\mysups{\mysups#1}}
\def\mydown#1{{\mytensor}_{\vphantom{\mysubs}#1}}
\def\myup#1{{\mytensor}^{\vphantom{\mysups}#1}}
\def\tensor#1#2{
  #1
  \def\mytensor{\vphantom{#1}}
  \def\mysubs{\relax}
  \def\mysups{\relax}
  \let\down=\mysavedown
  \let\up=\mysaveup
  #2
  \let\down=\mydown
  \let\up=\myup
  #2
  }

\newcommand{\C}{\mathbb C}
\newcommand{\R}{\mathbb R}

\renewcommand{\setminus}{\smallsetminus}

\renewcommand{\to}{\rightarrow}

\renewcommand{\epsilon}{\varepsilon}
\renewcommand{\hat}{\widehat}

\def\crn#1#2{{\vcenter{\vbox{
        \hbox{\kern#2pt \vrule width.#2pt height#1pt
           }
          \hrule height.#2pt}}}}

\newcommand{\newF}{\lambda}

\renewcommand{\hbar}{{\overline h}}

\newcommand{\pre}[2]{{{\vphantom{#2}}^{#1}}\kern-.2ex{#2}}

\theoremstyle{plain}
\newtheorem{theorem}{\sc Theorem}[section]
\newtheorem{assumptions}[theorem]{\sc Assumptions}

\newtheorem{proposition}[theorem]{\sc Proposition}
\newtheorem{corollary}[theorem] {\sc Corollary}
\newtheorem{question}[theorem] {\sc Question}

\newtheorem*{proposition*}{Proposition}
\newtheorem*{theorem*}{Theorem}
\newtheorem*{lemma*}{Lemma}
\newtheorem*{corollary*}{Corollary}

\theoremstyle{definition}

\newtheorem{Remark}[theorem]{\sc  Remark\rm}
\newtheorem{remark}[theorem]{\sc  Remark\rm}

\numberwithin{equation}{section}

\date{\today}

\renewcommand{\blue}[1]{#1}
\renewcommand{\red}[1]{#1}

\begin{document}

\title[Stationary  spacetimes with negative $\Lambda$] {Non-singular spacetimes with a negative cosmological constant:  V.  Boson stars}
\thanks{Preprint UWThPh-2017-23}

\author[Chru\'sciel]{Piotr T.~Chru\'sciel}

\address{Piotr
T.~Chru\'sciel, Erwin Schr\"odinger Institute and Faculty of Physics, University of Vienna, Boltzmanngasse 5, A1090 Wien, Austria}
\email{piotr.chrusciel@univie.ac.at} \urladdr{http://homepage.univie.ac.at/piotr.chrusciel/}

\author[Delay]{Erwann
Delay} \address{Erwann Delay, Universit\'e d'Avignon, Laboratoire de Math\'ematiques d'Avignon (EA 2151), 301 rue Baruch de Spinoza,
F-84916 Avignon, France}
\email{Erwann.Delay@univ-avignon.fr}
\urladdr{http://www.math.univ-avignon.fr/}

\author[Klinger]{Paul Klinger}

\address{Paul Klinger,  Faculty of Physics and Erwin Schr\"odinger Institute, University of Vienna, Boltzmanngasse 5, A1090 Wien, Austria}
\email{paul.klinger@univie.ac.at}

\author[Kriegl]{with an Appendix by Andreas Kriegl}
\author[Michor]{Peter W. Michor}
\author[Rainer]{Armin Rainer}
  \address{A.~Kriegl, P.W.~Michor, A.~Rainer: Fakult\"at f\"ur Mathematik, Universit\"at Wien,
  Oskar-Morgenstern-Platz~1, A-1090 Wien, Austria}
  \email{andreas.kriegl@univie.ac.at}

  \email{peter.michor@univie.ac.at}

  \email{armin.rainer@univie.ac.at}

\begin{abstract}
We prove existence of large families of solutions of Einstein-complex scalar field equations with a negative cosmological constant, with a stationary or static metric and a time-periodic complex scalar field.
\end{abstract}

\maketitle

\tableofcontents
\section{Introduction}
 \label{section:intro}

There is currently considerable interest in the literature in space-times with a negative cosmological constant $\Lambda$.
As a contribution to this, in two recent papers~\cite{ChDelayEM,ChDelayKlinger} we have provided proofs of existence of infinite dimensional families of  non-singular \emph{strictly stationary}
space times, solutions of the Einstein  equations with a negative cosmological constant and with various matter sources. The families of solutions constructed in~\cite{ChDelayKlinger}   include stationary metrics with a time-periodic complex scalar field $\scalarfield$, which are often referred to as \emph{boson stars}.  (By ``strictly stationary'' we mean that the Killing vector is timelike everywhere.) The Einstein-complex scalar field solutions we constructed can (but do not need to) have the usual AdS conformal structure at conformal infinity, and are driven by the asymptotic value of the scalar field, after requiring that the scalar-field potential $ G(|\scalarfield |^2)$ satisfies 
\bel{5XII16.2+}
    - {n^2} <4 \underbrace{\frac{ n(n-1)}{2|\Lambda|}}_{=:\ell^2}\Vpp < 0
 \,,
\ee
in space-time dimension $n+1$.
This might be seen as undesirable, since massive or massless linear scalar fields do not fulfill \eq{5XII16.2+}. The object of this work is to show how to modify the arguments in~\cite{ChDelayKlinger} to construct non-trivial boson star solutions for linear or nonlinear scalar fields with finite total energy of the field and with
``mass-squared parameter''%
\footnote{In the linear case, in which we have $G'(|\scalarfield |^2)\equiv \Vpp$ for all $\scalarfield$, when $\Vpp\ge0 $ the parameter  $\Vpp$ is usually identified with the square of the mass of the field.}
$\Vpp $ in the range \eq{5XII16.2+a}
and, if desired, usual conformal structure at the conformal boundary at infinity.

For definiteness we consider the Einstein equations involving a complex scalar field with a potential $ G(|\scalarfield |^2)$, so that the equation to be satisfied by $\scalarfield $ reads
\bel{17V17.11}
   \nabla^\mu  \nabla_\mu  \scalarfield  -  G'(|\scalarfield |^2)
    \scalarfield = 0
 \,,
\ee
and we note that the arguments here easily extend to include  contributions from further matter models as in \cite{ChDelayKlinger}.
We prove the following, where we normalise $\Lambda$ as in \eq{1VII17.1} below:

\begin{theorem}
 \label{T29VI17.1}
Let
$$
 \mathring \glorentz:=-\mathring V^2 dt^2
 + \mathring g_{ij}dx^i dx^j
$$
be a static, vacuum,  $C^{2}$-conformally compactifiable $(n+1)$-dimensional metric with $\mathring V>0$ such that
the associated operator $\Delta_L+2n$ (see \eq{30VI17.1} below) has no kernel in $L^2$. Let $\mathring\linpsi\not   \equiv0$ and $\mathring \omega\in \R^*$
solve the eigenvalue equation \eq{8I17.1}, and assume that the associated eigenspace is one-dimensional.
If
\bel{5XII16.2+a}
    - {n^2} < 4\Vpp
 \,,
\ee
then  for all time-independent $ \mathring \theta_i dx^i$ small enough
in $C^{2,\alpha}_1$ and  for all $\elambda\in\C$ with modulus small enough
there exists a time-independent metric
$$-V^2(dt+\theta_idx^i)^2 + g_{ij}dx^i dx^j
$$
near to and asymptotic to
 $ -\mathring V^2
 (dt + \mathring \theta_i dx^i)^2
 + \mathring g_{ij}dx^i dx^j$, solution of the Einstein-complex scalar field equations with $\scalarfield $ of the form
%
\bel{3II17.1}
 \scalarfield (t,x) = \elambda 
 e ^{i  \omega  t  } \chi(  x)
 \,,
\ee
where $x$ denotes space-variables, with $\chi$ decaying to zero at the conformal boundary, with   $\omega$ close to $\mathring \omega$
and $\chi$ close to $\mathring \linpsi$.
\end{theorem}

The proof of Theorem~\ref{T29VI17.1} is to be found in Section~\ref{s11VII17.1}.

The asymptotic behaviour of the solutions can be described precisely, see Remark~\ref{R28VII17.1} and  compare Section 7 of  \cite{ChDelayKlinger}. The solutions have complete asymptotic expansions in terms of (possibly non-integer) powers of $\rho$ and $\ln \rho$,
 with the following behaviour
in local coordinates near the conformal boundary $\partial M$
\bel{result0} V=\mathring
 V+o(\rho^{-1})\,,\quad \theta_i=\mathring \theta_i +o(1) ,\quad
 g_{ij}=\mathring g_{ij} +o(\rho^{-2})\,,
\ee
where $\rho$ is a coordinate which vanishes precisely at the conformal boundary. In fact, the deviation of the metric from the corresponding vacuum solution is determined by the asymptotic behaviour of the scalar field, which in the current case is
\bel{28VII17.1a}
 \scalarfield  = O(\rho^{(n+\sqrt{4 \ell^2 \Vpp + n^2})/2})
  \,.
\ee
This should be contrasted with the boson stars constructed in \cite{ChDelayKlinger}, where any small frequency $\omega$ is allowed, but
\bel{28VII17.1b}
 \scalarfield  = O(\rho^{(n-\sqrt{4 \ell^2 \Vpp + n^2} )/2})
  \,.
\ee
(The decay rates \eqref{28VII17.1a}-\eqref{28VII17.1b} are  sometimes associated with \emph{Dirichlet} or with \emph{Neumann} boundary conditions on $\scalarfield$  in the physics literature. From this point of view our solutions have zero Dirichlet data.) 
The requirement that $\scalarfield$ tends to zero at the conformal boundary led to the already-mentioned restriction \eqref{5XII16.2+} on the potential, which does not arise with the asymptotics \eqref{28VII17.1a}.

The energy-momentum tensor decays as
\bel{18XII16.11+}
           O(\rho^{n+\sqrt{4 \ell^2 \Vpp + n^2}})
            \,,
\ee
which gives a finite total energy of the scalar field.  

The condition on the kernel of $\Delta_L+2n$ is satisfied by the anti-de Sitter metric, or by small perturbations thereof constructed in~\cite{ACD2,ACD,ChDelayStationary}. We solve explicitly the eigenvalue equation for $\mathring \linpsi$ in this case, and check that some solutions have the required properties.

In Section~\ref{ss11II17.1} we establish a similar existence result near the anti-de Sitter metric for solutions of the form
\bel{11II17.1b}
 \scalarfield(r,\theta,\bvarphi ) =  \elambda e ^{i (\omega(\sigma) t-\emm \bvarphi)} \chi(\sigma,r,\theta)
 \,,
\ee
with $\Z\ni \emm\ne 0$, with $V$, $g$ and $\theta$ invariant under rotations of the azimuthal angle $\bvarphi $, with $\chi$ near $\mathring \linpsi\not\equiv 0$.

Our analysis is based on the observation that time-periodic solutions which decay to zero at infinity can be associated with eigenfunctions of the elliptic operator resulting from the linearisation of the scalar field equation.  The associated eigenvalues determine the allowed frequencies. We use eigenfunctions in anti-de Sitter space-time as a seed to construct one-parameter families of solutions.  We prove compactness of the associated resolvent, which guarantees a discrete frequency  spectrum at constant $\elambda$ in whole generality.

 In an appendix we generalise a result in \cite{KrieglMichor2} to establish differentiability of eigenvalues and eigenfunctions with respect to the metric, which allows us to use the implicit function theorem to prove existence of the desired solutions.

We note that our methods do not seem to easily generalise to black hole configurations with a periodic complex scalar field and stationary geometry, as constructed numerically in~\cite{DiasHorowitzSantos}.

The reader is referred to~\cite{Kaup,BizonWasserman,ChodoshSR,HerdeiroRaduKerr,LieblingPalenzuela} for boson stars and black holes with $\Lambda=0$.

 Once this work was finished we have been informed that the solutions, the existence of which has been proved here,  have been numerically constructed in~\cite{BrihayeHarmannRiedel}. We note that the numerical results there provide evidence for existence of ``large solutions'', beyond the implicit-function-theorem regime considered here.

\section{Properties of solutions near the anti-de Sitter metric: a summary}

The results here give a rigorous proof of existence of the near-AdS subset of the numerical solutions of~\cite{AstefaneseiRadu}, and provide many new near-AdS boson stars in all dimensions.
The resulting solutions near the anti-de Sitter metrics form a countable family of one-parameter solutions with finite total energy. More precisely, there is a discrete family of frequencies $\mathring \omega_{k,K}$ (cf.\ \eq{2II17.11} below), $k,K\in \N$,
tending to infinity as $k$ and $K$ tend  to infinity, such that, for each given sufficiently small $
\mathring \theta$ (possibly zero):

\begin{enumerate}
  \item Near the lowest frequency $\mathring \omega_{0,0}$ we obtain a family of solutions  parameterised by $\elambda$.
  \item When $\mathring \theta$ is spherically symmetric, near  each higher frequency $\mathring \omega_{0, K}$, $K\in \N$, we obtain a family of spherically symmetric solutions parameterised by $\elambda$.  The metrics are time-independent, and are static
      if $|\mathring V\mathring \theta|_{\mathring g} $ vanishes at the conformal boundary at infinity.
  \item When $\mathring \theta$ is axially symmetric, near each higher frequency $\mathring \omega_{k,0}$ and $\mathring \omega_{k,1}$  we obtain a family of stationary axially symmetric metrics parameterised by $\elambda$, with a field of the form \eqref{11II17.1b} with $|\emm|=k$   in dimensions $n\ge 3$,  with the value $|\emm|=k-1$ also being allowed in dimension $n=3$.
\end{enumerate}

In this list $|\elambda|$ is of course assumed to be small.
 Keeping in mind that we are assuming in this section that $-\mathring V^2 dt^2 + \mathring g$ is the anti-de Sitter metric,
  the  solutions are 
uniquely determined by $|\elambda|$ and the restriction of $ \mathring \theta$ to the conformal boundary, after compensating the phase of $\elambda$ by a shift of the time variable. Uniqueness holds within the class of solutions near the anti-de Sitter solution, as guaranteed by the implicit function theorem.

\section{Boson stars}

We recall that a Riemannian manifold $(M,g)$ is $C^2$ conformally compactifiable if $M$ is the interior of a smooth compact manifold $\bar M$ with non-empty boundary $\partial M:= \bar M\setminus M$, and if there exists a smooth function $\Omega$ on $\bar M$ which vanishes precisely on the boundary $\partial M$, with no critical points there, such that $\Omega^{-2}g$ extends to a $C^2$ Riemannian metric on $\bar M$.

\subsection{The argument}
 \label{s11VII17.1}
 
We wish to prove existence of continuous families of boson stars, i.e., solutions of the Einstein-scalar field equations with a complex scalar field $\scalarfield $ of the form \eq{3II17.1} where
$\elambda$ is a complex constant varying over a neighborhood of the origin in the complex plane.
The metrics  we construct will take the form
\beal{gme1} &\glorentz  = -V^2(dt+\underbrace{\theta_i
dx^i}_{=\theta})^2 + \underbrace{g_{ij}dx^i dx^j}_{=g}\,, & \\ &
\partial_t V = \partial_t \theta = \partial_t g=0
\eeal{gme2}
(thus, $\mcL_{\partial_t}V= \mcL_{\partial_t}\theta = \mcL_{\partial_t} g=0$),
and will be near a metric of the form
\bel{17V17.2}
 \mathring \glorentz:= -\mathring V^2 dt^2 + \mathring g \equiv -\mathring V^2 dt^2 + \mathring V^{-2} dr^2 + r^2 \underbrace{\mathring h_{AB} (x^c) dx^A dx^B} _{=: \mathring h}
 \,,
\ee
where $\mathring V $ depends only upon $r$. We note, however, that the construction below works near any stationary solution satisfying the non-degeneracy and one-dimensional-kernel conditions spelled-out in Theorem~\ref{T29VI17.1}
whenever a non-trivial seed pair $(\mathring \psi,\mathring \omega)$ solving \eqref{matterequations2_periodicscalar} at $\mathring \glorentz$ is available.

Let $\Delta_L$ be
the  Lichnerowicz Laplacian  acting on  symmetric
two-tensor fields $u$,  defined
as~\cite[\S~1.143]{besse:einstein}
\bel{30VI17.1}
 \Delta_L
  u_{ij}=-D^kD_ku_{ij}+R_{ik}u^k{_j}+R_{jk}u^k{_i}-2R_{ikjl}u^{kl}\,.
\ee
We will assume that the operator $\Delta_L + 2n$
associated with the Riemannian metric
$$
 \mathring\griem:= \mathring V^2 dt^2 + \mathring g
$$
has no $L^2$-kernel; in such cases the metric $\mathring\griem$ is called \emph{non-degenerate}. Large classes of non-degenerate Einstein
metrics are described in~\cite{Lee:fredholm,ACD2,mand1,mand2}.

Substituting \eq{3II17.1} into \eq{17V17.11}, one obtains
\begin{equation}\label{26VII17.1}
\renewcommand{\arraystretch}{2}
\begin{split}
&
\big(
V D_k (Vg^{kj}\partial_j)
	- V^2 G'(|\elambda\chi|^2)\big)\chi
\\
&\qquad
+(1 - V^{2} \theta_k\theta^k) \omega^2 \chi
-
i\omega V^2 (\theta^j \partial_j\chi+V^{-1}D_j(V\theta^j \chi))=0\,.
\end{split}
\end{equation}
Denoting by $\linpsi$ the linearised counterpart of $\chi$, the linearisation of \eqref{26VII17.1} reads%
\footnote{Our notation is:  $\scalarfield $ for the possibly nonlinear scalar field, $\chi$ for the time-independent part thereof, $\linphi$ for the linearised scalar field,  and finally $\linpsi$ for the linearised counterpart of $\chi$.}
\begin{equation}\label{matterequations2_periodicscalar}
\renewcommand{\arraystretch}{2}
\begin{split}
&
 \big(
 \underbrace{V D_k (Vg^{kj}\partial_j
 )
 - V^2 \red{G'(0)}
 }_{=:\hat P}\big) \linpsi
\\
 &\qquad
 +(1 - V^{2} \theta_k\theta^k) \omega^2 \linpsi
 -
 i\omega V^2 (\theta^j \partial_j\linpsi+V^{-1}D_j(V\theta^j \linpsi))=0\,.
 \end{split}
\end{equation}

\medskip

We are ready now to pass to the

 \medskip

{\noindent \sc Proof of Theorem~\ref{T29VI17.1}:}
We show below 
that, under our hypotheses, \emph{there exists a differentiable map} $(\chi,\omega)$,
\bel{3II17.4a}
 (g,V,\theta)\mapsto \big(\chi=\chi(g,V,\theta),\omega=\omega(g,V,\theta)\big)
 \,,
\ee
which to triples $(g,V,\theta)$ near $(\mathring g, \mathring V,0) $ assigns a non-trivial solution to
 \eqref{26VII17.1}.
 For $\elambda$ and $\mathring \theta$ small enough  one can then use the implicit function theorem near $(\elambda,g,V,\red{\theta})=(0,\mathring g, \mathring V,\red{0})$ to solve the remaining Einstein-scalar field equations
\begin{equation}\label{mainequations_EMD_periodicscalar}
\renewcommand{\arraystretch}{2}
\left\{\begin{array}{l}
V(-\Delta_g V + n V) =  -\frac{1}{2}\omega^2|\elambda \chi|^2   + G'(|\elambda \chi|^2) \frac{V^2 |\elambda \chi|^2}{n-1} \, ,
\\
\begin{split}
    R_{ij}+n g_{ij}-& V^{-1}\nablariem _i\nablariem _j V=
     \frac{1}{2V^{2}} \lambda_{ik} \lambda^k{}_j
    + \frac{1}{2}\Re(\partial_i \elambda \chi \partial_j \overline {\elambda \chi})\\
    &+\frac{1}{2}(\theta_i\theta_j \omega^2 |\elambda \chi|^2 - \omega\theta_i \Im(\overline {\elambda \chi}\partial_j\elambda \chi)-\omega\theta_j\Im(\overline {\elambda \chi}\partial_i \elambda \chi))\\
    & +\frac{g_{ij}}{n-1} G'(|\elambda \chi|^2) |\elambda \chi|^2 \, ,
\end{split}
\\
V^{-1}\nablariem ^j (V \newF_{ij})= \omega\Im(\overline {\elambda \chi} \partial_i\elambda \chi)-\theta_i\omega^2 |\elambda \chi|^2 \, ,
\end{array}\right.
\end{equation}
where $\Im$ denotes  the imaginary part, with
$$
 \newF_{ij}=-V^2(\partial_i \theta_j - \partial_j
\theta_i)
 \, .
$$
We have also normalised the cosmological constant $\Lambda$ to
\bel{1VII17.1}
 \Lambda=-\frac{n(n-1)}{2}
  \,.
\ee
In \eq{mainequations_EMD_periodicscalar} the field $\chi$ and the frequency $\omega$ are understood as functions of $(g,V,\theta)$ given by the map \eq{3II17.4a}.
The reader is referred to the accompanying papers~\cite{ChDelayEM,ChDelayKlinger} for analytical details concerning \eqref{mainequations_EMD_periodicscalar}.
We simply mention that  the behaviour of the metric functions in local coordinates near the conformal boundary is given by \eq{result0}.
A precise set of
weighted function spaces that can be used when invoking the implicit function theorem in the proof are, e.g.,
\beal{result2}
 & V-\mathring V\in
     C_{1}^{k+2,\alpha}
 \,,\;\;
  g-\mathring g\in
    C_{2}^{k+2,\alpha} 
    \,,
 &
\\
 &
 \quad
 \theta
  -\mathring \theta
  \in C_{2}^{k+2,\alpha}
  \,,
 &
\eeal{result3}
for some $k\ge 0$ and $\alpha \in (0,1)$, with small norms in those spaces. Here, as elsewhere, $\rho$ is a defining function for the conformal boundary at infinity, with $\mathring V$ behaving as $\rho^{-1}$ for small $\rho$, and with $\rho^2 \mathring g$ extending smoothly to a positive-definite tensor field at $\rho=0$.
Finally, a tensor field $u$ is in $C^{k,\alpha}_\sigma$ if and only if $\rho^{-\sigma} u$ belongs to the usual H\"older space $C^{k,\alpha}( M)$, the norm being  defined
with respect to a fixed conformally compact metric on $M$, say $\mathring g$, in \eq{19VII17.1} below.

Alternatively, one can carry-out the proof in $L^2$-type Sobolev spaces, weighted to reflect the asymptotic behaviours just described, which are Hilbert spaces.
\qed

\medskip

\begin{Remark}
  \label{R28VII17.1}
When analysing the spectral properties of the operators associated with the equations satisfied by the scalar field $\scalarfield $, and thus $\chi$, one works in a subspace of the Hilbert space $H_{\delta'}^2$, with $\delta'$ restricted as in the proof of Proposition~\ref{propresoP}. It follows, however, from the equations at hand that the solutions satisfy
\bel{28VII17.1}
 \scalarfield  = O(\rho^{(\sqrt{4 \ell^2 \Vpp + n^2}+n)/2})
  \,,
\ee
and in fact $\scalarfield $ will be in a weighted H\"older space with this decay weight and  differentiability class as high as desired.
\end{Remark}

\begin{Remark}
  \label{R9VII17.1}
The above argument applies to real $\chi$'s and $\elambda$'s, with  $\mathring \theta\equiv 0$, in which case static metrics are obtained.
\end{Remark}

It thus remains to show the existence of non-trivial solutions of  \eqref{26VII17.1}  
such that the map \eqref{3II17.4a} exists and is differentiable.

\subsection{A spectral problem}
 \label{ss8I16.1}

The aim of this section is to establish Corollary~\ref{C5VII17.1} below, which guarantees a discrete set of frequencies for the problem at hand.

We start by considering  the $\linpsi$-equation  \eq{matterequations2_periodicscalar}
with $\theta\equiv 0$:
\begin{equation}
 \label{8I17.1}
P\linpsi:= \frac{V}{ \sqrt{\det g}}\partial_i (V \sqrt{\det g}g^{ij}\partial_j \linpsi) =(-    \omega^2 \pVpp) \linpsi
 \,,
 \quad
 \omega \in \R^*
 \,.
\end{equation}
It is convenient to get rid of the first-order-derivative terms in $P$. For this we define
\bel{4VII13.13}
 \psi=V^{-1/2}u
 \,.
\ee
Using $P=
V^2\Delta_g+V\spacenabla^iV\spacenabla_i$, we find
%
\beal{8III17.12}
 P\psi
 & = &
  V^{3/2}\left(\Delta_g+\frac{1}{4}V^{-2}|dV|^2_{g}-\frac 12 V^{-1} \Delta_g V\right)u
\\
 \nn
 & = &
  V^{3/2}\left(\Delta_g+\frac{1}{4}V^{-2}|dV|^2_{g}-\frac n2\right)u
	\\
 \nn
 & =: &
  V^{-1/2}V^2\big(\underbrace{\Delta_g+\mathcal V}_{=: P_{\mathcal V}} \big)u
   \,.
\eea
So $\psi$ is an eigenfunction of
\bel{3VII17.1}
  \hat P:= P - V^2 \Vpp
   \,.
\ee
if and only if $u$ is an eigenfunction of $V^2 (P_{\mathcal V}-\Vpp)$,
with the same eigenvalue. We wish to show that $\hat P$ has compact resolvent. This might seem surprising at first, as it is well known that for asymptotically hyperbolic metrics the operator $\Delta_g$ does not have a compact resolvent when acting, e.g., on $L^2(d\mu_g)$. The rationale for compactness in our case is that the principal symbol of our operator contains a multiplicative factor $V^2$ in front of $\Delta_g$, with $V$ going to infinity as the boundary is approached.

Let us set
$$
 L:= V^2(P_{\mathcal V}-\Vpp+\lambda V^{-2})\,.
$$
We will show that $L$, and then $\hat P$, have  discrete spectra when acting
on suitable weighted $L^2$-spaces.
Similarly to~\cite{Lee:fredholm} and~\cite{andersson:elliptic}, for $\delta\in \R$ we denote by $L^2_\delta$ the space  of functions (or tensors)
 $u\in L^2_{\loc}$ for which the following norm is finite:
$$
 \|u\|_{L^2_\delta}^2=\int_M|u|^2_g\rho^{2\delta}d\mu_g
$$
(recall that $\rho$ is a defining function for the conformal boundary at infinity). Similarly, for $k\in N$ and $\delta\in \R$, we denote  by $H^k_\delta$ the space  of functions (or tensors) in $H^k_{\loc}$, having covariant derivatives  up to order $k$ in $L^2_\delta$, with the obvious norms.

A triple $(M,V,g)$ will be said \emph{asymptotically hyperbolic} if
$$
 (S^1\times M,V^2dt^2+g)
$$
is
$C^2$-conformally compactifiable, with sectional curvatures approaching minus one at the conformal boundary. In this work we always assume  $V>0$.  The Riemannian counterparts $ V^2dt^2 +g$ of the $(n+1)$-dimensional AdS metrics $-V^2dt^2 +g$ provide examples of asymptotically hyperbolic metrics.

We have:

\begin{proposition}\label{propresoP}
Let $(M,g,V)$ be conformally compact and asymptotically hyperbolic.
For $\Vpp>-n^2/4$ set $s:=\sqrt{4\Vpp+n^2}$. The operator $V^2(P_{\mathcal V}-\Vpp)$ has  compact
resolvent when acting  on  $L^2_{\delta}$ if
$\delta^2<s^2/4$. In particular it has a discrete spectrum on this space.
\end{proposition}

\begin{proof}
For the proof it is useful to keep in mind that the asymptotic behaviour of the kernel of
$$
 \hat P + \lambda =  V^{2}\big(
 V^{-\frac12}(P_{\mathcal V}-\Vpp+\lambda V^{-2})V^{\frac12}\big)
$$
is governed by the characteristic indices of $P_{\mathcal V}-\Vpp+\lambda V^{-2}$,
with a further shift by $1/2$ due to the $V^{1/2}$ factor. The indices for $P_{\mathcal V}-\Vpp+\lambda V^{-2}$ are the solutions $\sigma_\pm$  of the equation
$$
 \sigma(n-1-\sigma)+\frac n2-\frac14+\Vpp=0
  \,,
$$
independently of $\lambda$,
that is
$$
\sigma_\pm=\frac{(n-1)\pm s}2
 \,.
$$
In either case the length of the characteristic interval for $\hat P + \lambda$ is $s$, centered at $n/2$.

To account for the above, and for the multiplication of $P_{\mathcal V}$ by $V^2\sim \rho^{-2}$ when passing to $\hat P$, we choose  $\delta'\in\R$ such that $\delta-2\leq\delta'<\delta$ and  $\delta'^2<s^2/4$.
In order to show that the resolvent  of our operator  is compact, we will prove
that its domain is a subset of  $H^2_{\delta'}$ and so is compactly embedded in
$L^2_{\delta}$ (cf., e,g.,~\cite[Lemma~3.6 (d)]{Lee:fredholm} or~\cite[Theorem~2.3 (6)]{andersson:elliptic}).

We first claim that  we have an inequality of the form:
\bel{4VI17.11}
 -\int_M u(P_{\mathcal V}-\Vpp)u\;d\mu_g\ge C(\epsilon) \int_M u^2d\mu_g,
\ee
for all $u$ smooth and compactly supported in $\{\rho<\epsilon\}$,  $C(\epsilon)$ being a positive constant.
For that, by, e.g., \cite[Lemma~3.4]{ChDelayEM}, there exists a constant $0<C<(n-1)^2/4$, which can be made as close to $(n-1)^2/4$
as desired by choosing $\epsilon$ small enough, such that for $u$ as above we have
$$
 \int_M |du|_g^2d\mu_g\ge C\int_M u^2d\mu_g
  \,.
$$
It holds that $\mathcal V\rightarrow  \frac14-\frac n2$
as conformal infinity is approached.
We deduce that \eq{4VI17.11}
 is satisfied if $\Vpp>-n^2/4$ with $C(\epsilon)$ as close as $\Vpp+n^2/4$ as desired.
We can thus use~\cite[Lemma~3.8]{andersson:elliptic} to conclude that
for $\delta^2<\Vpp+n^2/4=s^2/4$ there exists a positive constant $C(\epsilon,\delta)$ such that
\bel{4VII17.12}
\|(P_{\mathcal V}-\Vpp)u\|_{L^2_\delta}\ge C(\epsilon,\delta) \|u\|_{L^2_\delta},
\ee
for all $u$ smooth and compactly supported in $\{\rho<\epsilon\}$.

Suppose, now, that  $V^2(P_{\mathcal V}-\Vpp)u$  and $u$ are
in $L^2_{\delta}$. Since $V$ behaves as $\rho^{-1}$ near the boundary, it holds that  $(P_{\mathcal V}-\Vpp)u$  is in $L^2_{\delta-2}\subset L^2_{\delta'}$.
It follows from \eq{4VII17.12} and, e.g., \cite[Proposition~6.5 (a)]{Lee:fredholm} that $u\in H^2_{\delta'}$.
\end{proof}

\begin{corollary}
 \label{C5VII17.1}
Let $(M,g,V)$ be conformally compact and asymptotically hyperbolic.
For $\Vpp>-n^2/4$ and
%
$$
\delta^2<\frac {s^2} 4 \equiv  \Vpp + \frac {n^2} 4
$$
the operator $\hat P$ has  compact
resolvent when acting  on  $L^2_{\delta-\frac12}$.  In particular it has a discrete spectrum on this space.
\end{corollary}

\begin{proof}
Since
$$
\hat P=V^{-\frac12}[V^2(P_{\mathcal V}-\Vpp)]V^{\frac12},
$$
we have  $V^{\frac12}\psi\in L^2_{\delta}$ (compare \eq{4VII13.13}) if and only if $\psi\in L^2_{\delta-\frac12}$.
\end{proof}
%
%
\begin{remark}
 \label{R7VII17.1}
The  function  $\rho^{\beta}$ is in $L^2_{\delta-\frac12}$ if and only if
$$
\beta+\delta-\frac12>\frac{n-1}2.
$$
So $\rho^{\frac{n-s}2}$ is \emph{not} in $L^2_{\delta-\frac12}$ if and only if $\delta\leq s/2$, and
 $\rho^{\frac{n+s}2}$ is in $L^2_{\delta-\frac12}$ if and only if $\delta>-s/2$.
\end{remark}

\subsection{Non-zero $\theta$, and/or nonlinear scalar fields}
 \label{s19VII17.1}

We start by noting that:

\begin{proposition} \label{P1VII17.1}
  For $\Vpp>-n^2/4$, the operator \eqref{matterequations2_periodicscalar} acting  on  $L^2_{\delta-1/2}$  with
$\delta^2<s^2/4$ has compact resolvent as long as $V(|\theta|_g+|D^i\theta_i|)$ is small enough in $L^\infty$.
\end{proposition}

\proof
The operator \eq{8I17.1}, perturbed by a sufficiently small  bounded operator  between  the spaces  above, will keep
the compact resolvent property. Indeed,
for a bounded operator $B: H \to H$ whose norm is small, if $\hat P$ has compact resolvent then the sum $\hat P + B$ still has compact resolvent, since
%
\begin{equation}\label{22VII17.1}
 (\hat P+ B -\mu)^{-1} = (\operatorname{Id} + (\hat P - \mu)^{-1} B)^{-1}(\hat P - \mu)^{-1}
 \,.
\end{equation}
%
See also \cite[IV.3.17]{Kato76} for a  more general result.

Equation (\ref{matterequations2_periodicscalar})  provides such an operator if $V(|\theta|_g+|D^i\theta_i|_g)$ is small enough in $L^\infty$.
\qed

 Note that in the setting above, for each $ x=(V,g) $ in an open set of such pairs the resolvent $R(x):=(\hat P(x)-z)\i:H\to H$ is a compact operator for $z$ in the resolvent
set of $\hat P(x)$.

It follows now from the Appendix
that, given a simple eigenvector ${\zchi}$ of $\hat P$ at $(\mathring g, \mathring V)$ with frequency $\zomega \ne 0$,
there exists a  differentiable map $(\linpsi,\omega)$,
\bel{3II17.4}
 (g,V)\mapsto \big(\chihere =\linpsi(g,V),\omega=\omega(g,V)\big)
 \,,
\ee
which to pairs $(g,V)$ near $(\mathring g, \mathring V) $ assigns a non-trivial solution to
\eqref{matterequations2_periodicscalar}. 
We wish to extend this map to triples $(g,V,\theta)$ near $(\mathring g, \mathring V,0)$. For this, it is convenient to replace \eq{matterequations2_periodicscalar} by a self-adjoint version thereof. As in \eq{4VII13.13} we set $
 \chihere =V^{-1/2}u$, and using \eq{8III17.12} we have the following equivalent forms of \eqref{26VII17.1}
\beal{19VII17.11}
&&
\\
 \nonumber
 &
  (\hat P  + \omega^2)\chihere  =
 \underbrace{
  V^{2} \left[\theta_k\theta^k \omega^2 \chihere
 +
 i\omega (\theta^j \partial_j\chihere +V^{-1}D_j(V\theta^j \chihere ))
\red{ +\left(G'(|\elambda \psi|^2)-G'(0)\right)\chihere}
 \right]}
 _
 {=:F(\theta,\omega,\elambda, \chihere) }
 &
\\
 \nonumber
&
 \Longleftrightarrow
 \qquad
 \big(\underbrace{\Delta_g+\mathcal V}_{=: \PV} \red{- G'(0)}+ \red{V^{-2}}\omega^2 \big)u =
 V^{-3/2}
 F(\theta,\omega, \elambda, V^{-1/2}u)
\,.
&
\eea%
 In the linear case the idea of the argument is as follows:
Let ${\zu}$ be a simple eigenfunction of $\PV$ with norm one and with eigenvalue $-\zomega^2 $. Write
$$
 u = {\zu} + \delta u
 \,,
  \qquad
  \omega=\zomega + \delta \omega
$$
where $\delta \red{u}$ is orthogonal to ${\zu }$. Project the equation \eqref{19VII17.11} into a component orthogonal to ${\zu }$, and another along ${\zu }$. The operator $\hat P + \zomega $ is an isomorphism on the space orthogonal to ${\zu }$, and the equation can therefore be solved for $\delta \red{u}$ as a differentiable function of $\omega$ and $\theta$ for all $(\omega,\theta)$ close enough to $(\zomega ,0)$. One can then use the implicit function theorem to solve the equation along ${\zu }$ for $\omega$.

 We note that the solution will depend differentiable upon $g$ and $V$ by the results of Appendix~\ref{AS11VII17.1}.

Let us pass to the details of the above. We denote by $\Pzc$ the $L^2$-orthogonal projection operator on ${\zu }$ and by $\Pzperp$ the $L^2$-orthogonal projection on $\zu ^\perp$, thus
$$
 \Pzc f = \left( \int_M \overline {{\zu }} f d\mu_g \right) {\zu }
 \,,
 \quad
  \Pzperp
   f  = f - \Pzc f
   \,.
$$
We obtain yet another equivalent form of \eqref{matterequations2_periodicscalar} by rewriting the second equation in \eqref{19VII17.11} as the following pair of equations
\begin{eqnarray}
   \label{22VIIl17.3}
 \phantom{xxxxx}
   ( \PV+ \zomega^2) \delta u  &
 =
  &
   (\zomega^2 - \omega^2) \delta u  + \Pzperp\big( V^{-3/2}
 F(\theta,\omega)(V^{-1/2}(\zu+\delta u) \big)
  \,,
\\
  \underbrace{\omega^2 - \zomega^2}_{ \delta \omega (2\zomega +\delta \omega)} &=& \int \overline{\zu} \big( V^{-3/2}
 F(\theta,\omega)(V^{-1/2}(\zu+\delta u) \big)
  d\mu_{\mathring g}
  \,.
   \label{22VIIl17.2}
\end{eqnarray}
We need to show that for real $\omega$'s the right-hand side of \eqref{22VIIl17.2} is real. For this, some auxiliary notation is useful. Let us write the second of equations \eqref{19VII17.11}
as $Ax=y$, where $A$ is self-adjoint. Write $x=x_\parallel + x_\perp$, where $ x_\parallel$ is parallel to $\zu$ and $ x_\perp$ is orthogonal to $\zu$, similarly for $y$. Then \eqref{22VIIl17.3}-\eqref{22VIIl17.2} read
$$
  A x_\parallel = y_\parallel
  \,,
  \qquad  A x_\perp = y_\perp
  \,.
$$
The right-hand side of \eq{22VIIl17.2} equals the scalar product of $\langle x_\parallel, y_\parallel\rangle$, so we need to show that this scalar product is real. Now,
$$
 \langle x , y \rangle =
 \langle x_\parallel, y_\parallel\rangle+
 \langle x_\perp, y_\perp\rangle =
 \langle x_\parallel, y_\parallel\rangle+
 \langle x_\perp, A x_\perp\rangle
 \,.
$$
Since $A$ is self-adjoint the last term is real, and thus $\langle x_\parallel, y_\parallel\rangle$ will be real if and only if $\langle x , y \rangle$ is real. This is indeed the case, which can be seen by the following calculation:
\begin{eqnarray*}
    \lefteqn{
  \Im \left(
      \int \overline{(\zu + \delta u)} \big( V^{-3/2}
 F(\theta,\omega)(V^{-1/2}(\zu+\delta u) \big)
  d\mu_{\mathring g}
   \right)
   }
   &&
\\
 \nonumber
 &&   =
     \Im \left(
     \int \underbrace{\overline u}_{V^{1/2}\overline \chihere }  V^{-3/2}
 F(\theta,\omega)(V^{-1/2} u \big)
  d\mu_{\mathring g}
  \right)
\\
\nonumber
 &&   =
     \Im \left(
     \int V^{-1 }\overline \chihere \big(
  V^{2} \theta_k\theta^k \omega^2 \chihere
 +
 i\omega V^2 (\theta^j \partial_j\chihere +V^{-1}D_j(V\theta^j \chihere )) \big)
  d\mu_{\mathring g}
  \right)
\\
\nonumber
 &&   =
     \omega \Re \left(
     \int  \overline \chihere \big(V\theta^j \partial_j\chihere + D_j(V\theta^j \chihere ) \big)
  d\mu_{\mathring g}
  \right)
  =
     \omega
     \int  D_j(V\theta^j  |\chihere |^2  )
  d\mu_{\mathring g}
\\
 &&
  =0
  \,.
\end{eqnarray*}

The following alternative argument applies for scalar fields with a nonlinear potential $\mcV (\scalarfield) = G(|\scalarfield|^2)$ with $G(0)=0$ and $G'(0)>-n^2/4$:
 Set
$$
\tilde  F(\theta,\omega, \elambda, u):=V^{-3/2}
 F(\theta,\omega, \elambda, V^{-1/2}u)
  \,.
$$
Let $k+2>n/2$ and consider the operator from
$$
(\delta u, \delta \omega, \elambda,  \theta) \in (H^{k+2}\cap \zu^\perp)\times \R\times \R\times C^{1,0}_1
$$
to
$( H^k\cap \zu^\perp)\times \R$, which sends $(\delta u,\delta \omega, \elambda, \theta)$ to a pair
\begin{eqnarray}
 \nonumber 
\lefteqn{
 \Big(
   \Pzperp\big(( \PV\red{-G'(0)}+ \red{V^{-2}}\omega^2) (\delta u)  -  \tilde
 F(\theta,\omega, \elambda, \zu  +\delta u )\big)
  \,,
  }
  &&
\\
\nonumber 
 &&
 \phantom{xxx} \int \overline{\zu} \tilde
 F(\theta,\omega,\elambda, \zu +\delta u)
  d\mu_{\mathring g}-( 2 \red{V^{-2}}\zomega \delta \omega+ \red{V^{-2}}\delta\omega^2)
  \Big)
  \,.
\end{eqnarray}
The  derivative in the first two variables at $(0,0,0,0)$ is the map which sends
$(\delta u,\delta\omega)$  to
\begin{eqnarray}
 \nonumber 
\Big(
   \Pzperp\big(( \PV\red{-G'(0)}+ \red{V^{-2}} \zomega^2) (\delta u)  \big)
  \,,
 \
   -2 \red{V^{-2}}\zomega \delta \omega
    \Big)
  \,.
\end{eqnarray}
 The analysis of the linear problem just carried out shows   that this is an isomorphism for $\zomega\neq0$ (in weighted Sobolev spaces as spelled out elsewhere in this paper, see e.g. \cite{Lee:fredholm} for the relevant analytical estimates),
and the implicit function theorem shows that for small $\theta$ and $\sigma$
there exist $(\delta u,\delta\omega)$ solving \eqref{19VII17.11}.

A positive answer to the following question would immediately extend our analysis to non-simple eigenvalues  (compare, however, Remark~\ref{R3VIII17.1}, Appendix~\ref{AS11VII17.1} below):

\begin{question}
 \label{T2VII17.1}
Let $\mathring \psi$ be an eigenvector of $\hat P $ with eigenvalue $\mathring \lambda$ with multiplicity larger than one. Is it true that, under possibly some further natural restrictive conditions, there exists a neighborhood $\mycal U$ of $(\mathring V, \mathring g)$ and a \underline{differentiable map}
\bel{2VII17.1}
 {\mycal U} \ni (V,g) \mapsto (\psi, \lambda)
\ee
which to  $(V,g)$ assigns an eigenvector $\psi$ of $\hat P$ with eigenvalue $\lambda$:
\bel{2VIII17.2}
 \hat P \psi = \lambda \psi
 \,.
\ee
\end{question}

\subsection{Explicit solutions}

In order to apply our technique we need to make sure that there exist non-trivial solutions of the eigenvalue problem \eqref{8I17.1}, with one-dimensional eigenspaces. The solutions in this section were previously found in \cite{Breitenlohner1982} for dimension $n=3$ and in \cite{Mezincescu1985} for general $n$. We rederive the results here, to make clear the relation with the functional analysis results of Section~\ref{ss8I16.1}.

We consider the operator $P$ for the $(n+1)$-dimensional anti-de Sitter metric,
\bel{8I17.2}
 \zglor \equiv - \zV^2 dt^2 + \zg
 := -\big(\frac{r^2 }{\ell^2} +1\big) dt^2
    + \frac{dr^2}{\frac{r^2 }{\ell^2} +1} + r^2 \zh
    \,,
\ee
which we denote by $\mathring P$.
(We have $\ell=1$ with the normalization $\Lambda=-\frac{n(n-1)}{2}$, used elsewhere in this work, but we do not impose this condition in this section.)
Thus
\begin{equation}
 \label{8I17.3}
\zP\linpsi  = \zV^2 \big(  r^{-(n-1)} \partial_r (\zV^{\red{2}} r^{n-1}\partial_r \linpsi) + r^{-2}\Delta_{\zh} \linpsi\big)
 \,.
\end{equation}
Let $\linpsi= \sum \linpsiell(r)  \varphiell  $
be the decomposition of $\linpsi$ into eigenfunctions $\varphiell $ of $\Delta_{\zh}$, $\Delta_\zh \varphiell   = \lambdaell \varphiell $, where $\lambdaell\in\{-k(k+n-2)\}_{k=0}^\infty$.
For each angular mode one is thus led to a radial operator
\begin{equation}
 \label{8I17.4-}
 \zV^2 \big(  r^{-(n-1)} \partial_r (\zV^2 r^{n-1}\partial_r \linpsiell  ) +  \lambdaell r^{-2}  
 \,.
\end{equation}

We start with the equation
\begin{equation}\label{07II17.1}
 \zV^2 \big(  r^{-(n-1)} \partial_r (\zV^2 r^{n-1}\partial_r \linpsiell  ) +  \lambdaell r^{-2}  \linpsiell  \big) = -\omega^2 \linpsiell + \Vpp \zV^2 \linpsiell
 \,.
\end{equation}
Defining a new variable $z:=-r^2/\ell^2$ gives
\begin{equation}
\partial^2_z \linpsiell+\left(\frac{n/2}{z}+\frac{1}{z-1}\right)\partial_z \linpsiell
-\frac{\lambdaell(z-1)+\ell^2 \omega^2 z+\ell^2 \Vpp z(z-1)}{4z^2(z-1)^2}\linpsiell=0
\end{equation}
which manifestly has three regular singular points, at $0$, $1$ and $\infty$. An equation of this type, i.e. a homogeneous, linear, second order ODE with three regular singular points, can be transformed to the hypergeometric equation
(see, e.g.,~\cite[4.3.1]{kristensson2010}): We define $\tlinpsik (z):=z^{-\beta_1} (z-1)^{-\beta_2}\linpsiell(z)$ where $\beta_1$, $\beta_2$ are roots of the indicial equation at $z=0$ and $z=1$ respectively, given by
\begin{equation}
\beta_1=\frac{1}{4}\left(\sqrt{(n-2)^2-4 \lambdaell}-n+2\right)\,,\qquad
\beta_2=-\frac{1}{2} \ell \omega\,.
\end{equation}
The equation for $\tlinpsik $ is now
\begin{equation}
 \label{2VII17.11}
  \begin{split}
z(z-1)\partial^2_z\tlinpsik +\left(\left(2+R/2-\ell\omega\right) z-1-R/2\right)\partial_z\tlinpsik &\\
-\frac{1}{4}\left(\lambdaell+n+\ell^2(\Vpp-\omega^2)+(2+R)(\ell \omega-1)\right)\tlinpsik &=0
 \,, \end{split}\end{equation}
where, for $\lambdaell = -k(k+n-2)$,
$$
 R=\sqrt{(n-2)^2-4\lambdaell}\red{=n-2+2k}
 \,.
$$
Equation~\eqref{2VII17.11} is explicitly of the hypergeometric form, as already observed in~\cite{IshibashiWaldIII}.
The solutions are, after expressing $\tlinpsik $ by $\linpsiell $ and $z$ by $r$,
\begin{equation}\label{12II17.5}
\begin{split}
\linpsiell(r) =(\ell^2+r^2)^{-\ell \omega/2}
\Big[&C_1r^{(2+R-n)/2} {}_2F_1\big(A_1, B_1, 1+R/2, -r^2/\ell^2\big)\\
&+C_2r^{(2-R-n)/2}  {}_2F_1\big(A_2, B_2, 1-R/2, -r^2/\ell^2\big)\Big]
 \,,
\end{split}
\end{equation}
where
\begin{align*}
&A_1:=\frac{1}{4}(-s+R+2-2 \ell w)\,,&B_1:=\frac{1}{4}(s+R+2-2\ell \omega)\,,\\
&A_2:=\frac{1}{4}(-s-R+2-2 \ell w)\,,&B_2:=\frac{1}{4}(s-R+2-2\ell \omega)\,,&
\end{align*}
$s=\sqrt{4 \ell^2 \Vpp + n^2}$, $C_1$, $C_2$ are constants and ${}_2F_1$ is the usual hypergeometric function:
\begin{equation}
 \label{12II17.1}
{}_2F_1(a,b,c,z)=\sum_{n=0}^{\infty} \frac{(a)_n(b)_n}{(c)_n}\frac{z^n}{n!}
\end{equation}
with
\begin{equation}
(a)_n=
\begin{cases}
1\,, & \text{for }n=0\,,\\
a(a+1)\hdots (a+n-1)
 \,, & \text{for } n>0\,.
\end{cases}
\end{equation}
Strictly speaking, \eq{12II17.1} holds
for $|z|<1$, and analytic continuation should be used otherwise.

Keeping in mind that $\lambdaell\le 0$, and since  ${}_2F_1(a,b,c,0)=1$ for all $a,b,c$, the solutions are regular at $r=0$ if and only if $C_2=0$. For $\lambdaell=0$ the solutions converge to $\ell^{-\ell \omega}$ at the origin, while for $\lambdaell<0$ they converge to zero there.

If $a$ or $b$ are non-positive integers then $_2F_1$ is a polynomial. We will show that these are the only solutions in our context.

 In the polynomial case the term ${}_2F_1(a,b,c,-r^2/\ell^2)$ behaves, for $r\to \infty$, as $r^{-2a}$ if $-a\in\N$, as $r^{-2b}$ if $-b\in \N$ and as $r^{-2\max(a,b)}$ if both $-a\in\N$ and $-b\in\N$.
In order to analyse the behavior of ${}_2F_1(a,b,c,z)$ as $|z|\to\infty$ in the general case, the difference $b-a$ is relevant:

Indeed, if $b-a$ is not an integer, then~\cite[Eq.~15.8.2]{NIST:DLMF}
\begin{equation}\label{02II17.2}\begin{split}
{}_2F_1(a,b,c,-r^2/\ell^2)=\frac{\pi}{\sin(\pi(b-a))}\Big[
&\frac{r^{-2a}\ell^{2a}}{\Gamma(b)\Gamma(c-a)}\left(1+O(r^{-2})\right)\\
&-\frac{r^{-2b}\ell^{2b}}{\Gamma(a)\Gamma(c-b)}\left(1+O(r^{-2})\right)
\Big]\,.
\end{split}\end{equation}

If $b-a$ is a non-negative integer, say $p$, then ${}_2F_1(a,b,c,-r^2/\ell^2)$ has the asymptotics~\cite[Eq.~15.8.8]{NIST:DLMF}
\begin{equation}\label{02II17.1}\begin{aligned}
{}_2F_1(a,b,c,-r^2/\ell^2)=&
\frac{r^{-2a}\ell^{2a}}{\Gamma(a+p)}\left(\frac{(p-1)!}{\Gamma(c-a)}\left(1+O(r^{-2})\right)\right)\\
&+\frac{r^{-2b}\ell^{2b}}{\Gamma(a)}\biggl(\frac{\ln(r^2/\ell^2)}{p!\, \Gamma(c-a-p)}+\frac{\Gamma'(a+p)}{\Gamma(a+p)}\\
&\phantom{+\frac{r^{-2b}\ell^{2b}}{\Gamma(a)}\biggl(}\,\,+\frac{\Gamma'(c-a-p)}{\Gamma(c-a-p)}+\Gamma'(1)+O(r^{-2})\biggr)\,.
\end{aligned}\end{equation}
If $b-a$ is a negative integer, \eqref{02II17.1} applies with $a$ and $b$ interchanged.

In our case, in the hypergeometric function in the $C_1$ term we have $\Re(b-a)=\Re(s/2)\geq 0$.

The question arises, how the asymptotic behaviour just seen fits with the analysis in Section~\ref{ss8I16.1}. To make contact with Corollary~\ref{C5VII17.1} we need to check when  $\linpsiell  \in L^2_{\delta+1/2}=L^2(r^{-2\delta -3} d\mu_{\mathring{g}})$.  In view of Remark~\ref{R7VII17.1} we see that choosing
\bel{7VII17.1}
   -\frac s2  < \delta < \frac s 2
    \,,
\ee
the operator $\hat P $ will have compact resolvent, and discrete spectrum.

With the choice \eq{7VII17.1} of the weighted spaces the $r^{-2a}$ term in both \eqref{02II17.2} and \eqref{02II17.1} does not decay fast enough. Working in weighted spaces with a weight as in \eqref{7VII17.1}, we thus need to ensure  that the coefficient in front of the $r^{-2a}$ term vanishes. Equivalently, either 1) $b=(s+R+2-2\ell \omega)/4$ needs to be a non-positive integer, say $-K$ with $K\in \N$, or 2) $c-a$ should be a non-positive integer $-K$.
In case 1), this gives a polynomial solution and $\red{r^{-(s+n)/2}}$  behavior for $\linpsiell$.

The bottom line is that the weighted-Sobolev space condition above will be satisfied in case 1) if and only if
\begin{equation}
 \label{2II17.11}
\omega=\frac{n+ \sqrt{4 \ell^2 \Vpp + n^2}+2(k + 2 K)}{2 \ell},
\quad\text{for}\quad
K=0,1,2,3,\hdots
\end{equation}
In case 2), one is similarly led to an overall minus sign at the right-hand side of \eq{2II17.11}, resulting in the same value of $\omega^2$.

The solutions obtained in both cases differ only by a constant factor: a change of $\omega\to -\omega$ in \eqref{12II17.5} leads to the same one-dimensional family of solutions, as follows from the identity~\cite[15.8.1]{NIST:DLMF}
\begin{equation}
{}_2F_1(a,b,c,z)=(1-z)^{c-a-b}{}_2F_1(c-a,c-b,c,z)\,.
\end{equation}

Different $\lambdaell$ can lead to the same eigenvalue: the pair $(k, K)$ gives the same value of $\omega^2$ as $(k', K+\frac{k-k'}{2})$ (for even $k-k'$).
	
Keeping in mind that our construction works best for simple eigenvalues, we note that only the eigenvalue with $k=0$ ($\lambda_0=0$) and $K=0$ is simple.

However, if we restrict the whole construction above to the space of spherically symmetric eigenfunctions and spherically symmetric static metrics, then all eigenvalues will be simple in this space, leading to a countable family, parameterised by $K$, of one-parameter families of solutions, each parameterised by $\elambda$, of the full system of equations.

\subsection{Rotating solutions}
 \label{ss11II17.1}

In this section we let $0\ne k\in \N$ and consider frequencies $\omega$ near \eqref{2II17.11}.
Instead of \eqref{3II17.1} we assume that
\bel{11II17.1a}
 \scalarfield(\elambda, r,\btheta^j,\bvarphi) = \elambda e ^{i (\omega (\elambda) t-\emm \bvarphi)} \chi(\elambda, r,\btheta^j)
 \,,
\ee
where 
we denote by $(\btheta^j,\bvarphi)$, $\btheta^j\in[0,\pi)$, $\bvarphi\in [0,2\pi)$ the usual angular coordinates on $S^{n-1}$. For definiteness we assume that $-\mathring V^2 dr^2 + \mathring g$ is the anti-de Sitter metric, and that $\mathring \theta =0$, and that $\chi$ will be near a given solution of the eigenvalue problem.
This will be consistent with \eqref{matterequations2_periodicscalar} if we restrict ourselves to a class of metrics, potentials $V$ and forms $\theta$  which are invariant under $\partial_\bvarphi$:
$$
 \mcL_{\partial_\bvarphi} g = 0\,,
  \quad
  \mcL_{\partial_\bvarphi} V =0
  \,,\quad
  \mcL_{\partial_\bvarphi} \theta =0
  \,.
$$
In the space of functions $\scalarfield $ of the form \eq{11II17.1a} the eigenspaces corresponding to $k=|m|$,
$K=0$ in \eqref{2II17.11} are one-dimensional: Indeed, as already noted above, the pair $(k,K)$ leads to the same eigenvalue $\omega$ as $(k', K')$ if and only if $K'=K+(k-k')/2$. As $K$ is required to be non-negative this means that $k'<k$ if $K=0$. But the eigenfunctions of the Laplacian on $S^{n-1}$ take the form (cf., e.g., \cite{Higuchi,Higuchi2})
\[
Y_{\ell_1,\dots,\ell_{n-2},k}(\btheta^j, \bvarphi)=e^{i\ell_1 \bvarphi}Z_{\ell_2,\dots,\ell_{n-2},k}(\btheta^j)\,,
\]
where $|\ell_1|\leq\ell_2\leq\dots...\leq k$, and therefore there are no eigenfunctions which behave as $e^{-i m \bvarphi}$ for $k'<k=|m|$.

In dimension $n=3$ the eigenspace corresponding to $k=|m|+1$, \red{$K=0$} is also one-dimensional: In this case the eigenfunctions of the Laplacian are parameterized by only $2$ parameters, $k$ and $\ell_1$, and therefore there is at most one eigenfunction of the form \eqref{11II17.1a} for a fixed $k$. The only value $k'<k$ with eigenfunctions of the required form is $k'=k-1=|m|$, which does not give the same $\omega$ for any integer value of $K'$.

 All our arguments can be repeated in this setting, leading to solutions of the coupled equations with suitable $\omega$ for all $\elambda$ in \eq{11II17.1a} small enough.

\appendix

\renewcommand{\be}{\beta}

\section{Smooth dependence upon $(V,g)$\\ by A. Kriegl, P. Michor and A. Rainer}
 \label{KMR}

\newcommand{\Vdomain}{\red{\mathfrak{V}}}
\renewcommand{\al}{\alpha}
 \label{AS11VII17.1}

Let $M$ be a non-compact smooth manifold (the Cauchy surface). Let $\red{\mathring g}$ be a smooth background
Riemannian metric on $M$ with good properties (at least complete, or with bounded geometry, for example).
Let $\rh\in C^\infty(M,\mathbb R_{>0})$ be a fixed smooth positive function on $M$.

Let $E\to M$ be a tensor bundle like  $S^2T^*M$. Then let
$\Ga_{C^{k,\al}(\red{\mathring g})}(E)$ be the Banach space of all $C^{k,\al}$ sections  $f$ of $E$ (for each smooth curve $c:\mathbb R\to M$ the composition $f\o c$ is $C^k$ and its $k$-th derivative is locally H\"older of class $\al$) with norm (one of many equivalent conditions)
\begin{align}
 \label{19VII17.1}
\|f\|_{C^{k,\al}} =&\sup_{x\in M}\Big( |f(x)| + |\nabla^{\red{\mathring g}}f(x)|_{\red{\mathring g}} + \dots + |(\nabla^{\red{\mathring g}} )^kf(x)|_{\red{\mathring g}}\Big)
\\&
+ \sup_{0<\on{dist}^{\red{\mathring g}}(x,y)\le \ep} \frac{\|(\nabla^{\red{\mathring g}})^kf(x)-\on{Pt}_{y,x}((\nabla^{\red{\mathring g}})^kf(y))\|_{\red{\mathring g}}}
{\on{dist}^{\red{\mathring g}}(x,y)^\al}
 \,,
 \nonumber
\end{align}
where we used the geodesic distance on $M$ induced by $\red{\mathring g}$, and the fiber metric on
$\otimes^kT^*M\otimes E$ induced by $\red{\mathring g}$, and the parallel transport $\on{Pt}_{y,x}$ from $y$ to $x$ along the short geodesic from $x$ to $y$; here $\ep$ is smaller than the injectivity radius of $(M,\red{\mathring g})$.

Below we will meet mappings on $\Ga_{C^{k,\al}(\red{\mathring g})}(E)$ which are real analytic; since we are on a Banach space, these are given by convergent power series of bounded multilinear homogeneous expressions. See \cite[Sections 10 and 11]{KrieglMichor}.
These mappings will be visibly real analytic, since they will involve only differentiations, multiplications, and inversion of matrices.

Moreover we let
$\Ga_{C^{k,\al}_r(\red{\mathring g})}(E) = \{s: \rh^{-r}s \in \Ga_{C^{k,\al}(\red{\mathring g})}(E)\}$.
If $E$ is the trivial line bundle we just write $C^{k,\al}_r(\red{\mathring g})(M)$ instead of
$\Ga_{C^{k,\al}_r(\red{\mathring g})}(M\x \mathbb R)$.

Let $\red{\mathring V}$ be a smooth positive function on $M$ and let $\mathcal V$ be the space of all  functions $V\in C^{\infty}(M,\mathbb R)$ such that $V-\red{\mathring V}\in C^{k+1,\al}_1(\red{\mathring g})(M)$ and $V>0$ everywhere on $M$.
Then $\mathcal V$ is an open set in an affine space modelled on a Banach space.

Let $\mathcal M$ be the space of all Riemannian metrics $g$ on $M$ such that
$g-\red{\mathring g}\in \Ga_{C^{k+2,\al}_2(\red{\mathring g})}(S^2T^*M)$.
It then follows that also
$g\i-\red{\mathring g}\i\in \Ga_{C^{k+2,\al}_2(\red{\mathring g})}(S^2T^*M)$.
Note that 
$\mathcal M$
is an open set in an affine space modelled on a Banach
space of tensor fields on $M$.

\noindent\thetag{1}
\emph{
For each $g\in \mathcal M$ the volume density
$\on{vol}(g)$ satisfies
$ \on{vol}(g) = F(g)\on{vol}(\red{\mathring g})$
for a function $F(g)= \sqrt{\frac{\det(g_{ij})}{\det(\red{\mathring g}_{ij})}}\in C^{k+2,\al}(\red{\mathring g})(M)$
with $F(g)>0$ and $\frac1{F(g)}\in C^{k+2,\al}(\red{\mathring g})(M)$. Moreover, $F(g)$ visibly depends real
analytically on $g$.
}

\noindent\thetag{2} \emph{ For each $g\in \mathcal M$ the Hilbert space
$L^2(\rh^{2\delta} \on{vol}(g))$ is isomorphic (but not isometric) to $L^2(\rh^{2\delta} \on{vol}(\red{\mathring g}))$ via the
multiplication operator $F(g):L^2(\rh^{2\delta} \on{vol}(\red{\mathring g}))\to L^2(\rh^{2\delta} \on{vol}(g))$. }
So we may consider just
one Hilbert space $L^2(\rh^{2\delta} \on{vol}(\red{\mathring g}))$ with different inner products $\langle \al,\be\rangle_g =
\langle F(g)\i\al,F(g)\i\be\rangle_{\red{\mathring g}}$.

We pass now to the description of our assumptions:

\begin{assumptions}
  \label{A5VII17.1}
Let $U:= \mathcal V\x  \mathcal M$, an open subset in an affine space modelled on a Banach space.
For each $x  \in U$,
let $A(x)$ be an
\red{unbounded closed operator 
	 on $L^2(\rh^{2\delta} \on{vol}(g))$}, such that the
domains satisfy $\mathcal D(A(x)) = F(g)\mathcal D(A(\red{\mathring V},\red{\mathring g}))$.

\end{assumptions}

In the case of interest in this paper $x = (V,g)$.

 The
operators in this paper even have equality of all domains in $L^2(\rh^{2\delta} \on{vol}(\red{\mathring g}))$.

By replacing $A(x)$ with $F(g)\i A(x)F(g)$ we may assume that
$A$ is a map from an open subset $U:=\mathcal V\x\mathcal M$ in an affine space modelled on a Banach space
to the set of unbounded closed operators on some fixed Hilbert space $H:=L^2(\rh^{2\delta} \on{vol}(\red{\mathring g}))$ with
common domain $\Vdomain =\mathcal D(A(x))\subseteq H$.
Furthermore, we assume that $A$
is real analytic in the sense,
 that for all vectors $v\in \Vdomain $ and $w\in H$
the composite $x\mapsto \langle A(x)v,w\rangle$ is real analytic; see \cite[Section 10]{KrieglMichor} for more information. This weak definition suffices due to the real analytic uniform boundedness theorem \cite[11.12]{KrieglMichor}.

We emphasise that it is not assumed that the $A(x)$'s are self-adjoint  or with compact resolvent.

\begin{theorem}
 \label{TA1VIII17.1}
	Under the assumptions \ref{A5VII17.1}, let $\la(x_0)$ be a simple \red{isolated} eigenvalue of $A(x_0)$
	with eigenvector $v(x_0)\in \Vdomain $, where $x_0 \in U$ is fixed.

	Then one may extend $\la$ and $v$ to locally defined real analytic  mappings, such that
	$\la(x)$ is a simple \red{isolated}
eigenvalue of $A(x)$ with corresponding eigenvector $v(x)$ for \red{$x$} near $x_0$
  	in $U$.
	
	If $\la$ is real-valued and non-negative (e.g.\ $A(x)$ is symmetric for some inner product possibly different from
  	the given one), then the non-negative root $\om(x)=\sqrt{\la(x)}$ is locally Lipschitz in \red{$x$}.
	On the subset of those $x$ for
	which $\la(x)>0$ the function $\om(x)$ depends real analytically  on $x$.
\end{theorem}
It is in general not possible to have a differentiable function $\om(x)$ such that $\om(x)^2 = \la(x)$,
see e.g.\ \cite[5.2]{AKLM98}.

\begin{proof}
The following argument is adapted from \cite[7.4]{AKLM98} and \cite[Proof of resolvent lemma]{KMRp}:
For each $x\in U$ consider the norm $\|u\|_x^2:=\|u\|_H^2+\|A(x)u\|_H^2$ on
$\Vdomain $. Since $A(x)$ is closed, $(\Vdomain ,\|\quad\|_x)$ is also a
Hilbert space with inner product
$\langle u,v\rangle_x:=\langle u,v\rangle_H+\langle A(x)u,A(x)v\rangle_H$.
Then $U\ni x\mapsto \langle u,v\rangle_x$ is real analytic for fixed $u,v\in \Vdomain $,
and by the multilinear uniform boundedness principle
\cite[5.18 and 11.14]{KrieglMichor}, the mapping
$x\mapsto \langle \;,\;\rangle_x$ is real analytic into the space of
bounded bilinear forms on $(\Vdomain , \|\quad\|_{x_0})$.
By the exponential law
\cite[3.12 and 11.18]{KrieglMichor} the mapping
$(x,u)\mapsto \|u\|^2_x$ is real analytic from
$U \times (\mathfrak V,\|~\|_{x_0})$ to $\mathbb R$
for each fixed $x_0$.
Thus, all Hilbert norms $\|\quad\|_x$ are equivalent: for  $B\subset U$ bounded,
$\{\|u\|_x:x\in B,\|u\|_{x_0}\le 1 \}$ is bounded by $C_{B,x_0}$ in
$\mathbb R$, so $\|u\|_x\le C_{B,x_0}\|u\|_{x_0}$ for all $x\in B$. Moreover,
each $A(x)$ is a globally defined operator $(\mathfrak V,\|~\|_{x_0}) \to H$
with closed graph and is thus bounded, and by using again the
(multi)linear uniform boundedness principle \cite[5.18 and 11.14]{KrieglMichor}
as above we see
that $x\mapsto A(x)$ is
real analytic $U \to L((\mathfrak V,\|~\|_{x_0}),H)$.
	
	We consider the global resolvent set
	\[
	\mathcal R = \{(x,\mu)\in U\x \mathbb C: A(x)-\mu:
	 (\Vdomain,\|\quad\|_{x_0}) \to H\text{ is invertible}\}
	\]
	which is an open subset of $U\x \mathbb C$, since
  $(A(x)-\mu)\o (A(x_0)-\mu_0)\i\in L(H)$ and equals  $\on{Id}$ for  $(x,\mu) = (x_0,\mu_0)$.
	By assumption, $\la(x_0)$ is a simple \red{isolated} eigenvalue of $A(x_0)$ with eigenvector $v(x_0)$.
	We choose a smooth 	positively oriented curve $\ga$ in $\mathbb C$ which contains only $\la(x_0)$ in its interior
	and all other eigenvalues of $A(x_0)$ in the exterior; in particular,
	$\{x_0\}\x \ga \subset \mathcal R$. \red{Since $\ga\subset \mathbb C$ is compact, we may cover $\{x_0\}\x \ga$ by finitely many open sets
	of the form  $W_i\x \tilde W_i$ contained in $\mathcal R$; for $U_1 = \bigcap W_i$ we then have $U_1\x \ga \subset \mathcal R$ where $U_1$ is an open neighborhood of $x_0$ in $U$.
	By \cite[III.6.17]{Kato76}, for $\red{x \in U}_1$ the spectrum of $A(x)$ is separated into the two parts contained in the interior} and in
  	the exterior of $\ga$ and the resolvent integral
	\[
	P(x) = -\frac{1}{2\pi i}\int_\ga (A(x) - \mu)\i \,d\mu: H\to \Vdomain  \subseteq H
	\]
	is a projection operator onto the sum of all generalized eigenspaces of all eigenvalues of $A(x)$ in the interior of $\ga$, for $\red{x \in U}_1$.
	We now argue as in the proof of \cite[7.8, Claim 1]{AKLM98} (see also \cite[50.16, Claim 1]{KrieglMichor}) as follows:
	By replacing $A(x)$ by $A(x)-z_0$ if necessary we may assume that 0 is not in the interior of $\ga$.
	Since $U_1\ni x\mapsto P(x)$ is a smooth (even real analytic) mapping
	into the space of bounded projections in $\red{L(H)}$ with finite dimensional ranges,
	the rank of $P(x)$ cannot fall locally, and it cannot increase locally since the distance in
	$L(H)$ of $P(x)$ to the subset of operators of rank 1
	is continuous in \red{$x$} and is either 0 or $\ge 1$.
	See also \cite[I.\S4.6 and I.6.36] {Kato76}.

	So we conclude that  for $x$ in a (possibly smaller) open set $U_1$
	there is only one (counted with multiplicity) eigenvalue (denoted $\la(x)$) of $A(x)$
	in the interior of $\ga$ and hence $P(x)$ is a projection on its eigenspace.
	See also \cite[IV.\S3.4-5]{Kato76}.
	
	
	Then $v(x): = P(x)v(x_0)$ is an eigenvector for $A(x)$ depending real analytically on $x$ near $x_0$. The
	corresponding eigenvalue is also real analytic, since
	\[
	\la(x) = \frac{\langle A(x)v(x),v(x)\rangle_{\red{\mathring g}}}{\|v(x)\|_{\red{\mathring g}}^2}\,.
	\]
Near positive $\lambda(x_0)$'s the square root is obviously also real analytic.
	If the smooth eigenvalue $\la(x)$ is always non-negative, then the non-negative square root $\om(x)=\sqrt{\la(x)}$
	is locally Lipschitz in $x$ by \cite{KMR}.
\end{proof}

\begin{Remark}\label{R3VIII17.1}
The assumption that $\lambda(x_0)$ is a simple eigenvalue of $A(x_0)$ is quite essential in the
above theorem. Near eigenvalues with higher multiplicity the situation becomes much more
difficult. Real analytic curves of self-adjoint or normal unbounded operators with compact resolvent and
common domain of definition admit real analytic choices of their eigenvalues and eigenvectors.
However, if the parameter space is at least $2$-dimensional,  examples can be given with no
differentiable choice for self-adjoint operators and no continuous choice for normal operators.
If the parameter space is finite dimensional, then, locally, the eigenvalues and
eigenvectors can be chosen real analytically after blowing up the parameter space.
Even less can be said if the operators depend only smoothly on a parameter and distinct eigenvalues
have infinite order of contact. Without normality even real analytic curves of diagonalisable matrices
need not admit smooth choices of  the eigenvalues. All this can be found in
\cite{RainerN}
and the references therein. For the optimal (Sobolev) regularity of the eigenvalues of smooth curves of
arbitrary quadratic matrices see~\cite{ParusinskiRainer15}.

\end{Remark}

 \bigskip

\noindent{\sc Acknowledgements} The research of PTC was supported in
part by the Austrian Research Fund (FWF), Project  P29517-N27,  by the Polish National Center of Science (NCN) under grant 2016/21/B/ST1/00940 and by the Erwin Schr\"odinger Institute. Armin Rainer was supported by the FWF-Project P~26735-N25. Paul Klinger was supported by a uni:docs grant of the University of Vienna. We are grateful to  Gilles Carron and Luc Nguyen for useful comments and discussions.

\def\polhk#1{\setbox0=\hbox{#1}{\ooalign{\hidewidth
  \lower1.5ex\hbox{`}\hidewidth\crcr\unhbox0}}} \def\cprime{$'$}
  \def\cprime{$'$}
\providecommand{\bysame}{\leavevmode\hbox to3em{\hrulefill}\thinspace}
\providecommand{\MR}{\relax\ifhmode\unskip\space\fi MR }
\providecommand{\MRhref}[2]{%
  \href{http://www.ams.org/mathscinet-getitem?mr=#1}{#2}
}
\providecommand{\href}[2]{#2}

\end{document}